\documentclass[twocolumn,
 amsmath,amssymb,
 aps,
 nofootinbib, 
]{revtex4-1}

\usepackage{graphicx}
\usepackage{subcaption}
\usepackage{caption}
\usepackage{dcolumn}
\usepackage{bm}
\usepackage{hyperref}
\usepackage{color} 
\usepackage{amsmath}
\usepackage[makeroom]{cancel}
\usepackage{times}
\usepackage{autonum}
\usepackage{graphicx}
\usepackage{dcolumn}
\usepackage{enumerate,amsthm,amssymb,color}
\usepackage{bbm}
\usepackage{bm}
\usepackage{tikz}
\usepackage{xurl} 
\usepackage[english]{babel}
\usepackage{synttree}
\usepackage{multirow}
\usepackage{xfrac}
\usepackage{lipsum}
\usepackage{natbib} 

\usepackage{physics}
\usepackage{dsfont}

\def\defeq{\mathrel{\mathop:}=} 

\definecolor{violet}{HTML}{53257F} 
\definecolor{green}{HTML}{257a7f}
\definecolor{green2}{HTML}{4D7F25}
\definecolor{brown}{HTML}{852e29}

\colorlet{nodecolor}{green}
\colorlet{edgecolor}{violet}
\colorlet{arrowcolor}{brown}
\colorlet{highlightcolor}{green2}
\colorlet{highlightcolor2}{brown}

\newtheorem{theorem}{Theorem}
\newtheorem{corollary}{Corollary}
\newtheorem{lemma}{Lemma}

\newcommand{\GHZ}[1]{\ensuremath{\ket{\mathrm{GHZ}}_{#1}}} 
\newcommand{\Lin}[1]{\ensuremath{\ket{\mathrm{L}}_{#1}}} 

\newcommand{\node}[1]{\ensuremath{#1}}

\newcommand{\network}[0]{\ensuremath{V_{L}}}
\newcommand{\measured}[0]{\ensuremath{V_{M}}}
\newcommand{\GHZnodes}{\ensuremath{V_{G}}}

\newcommand{\GHZsize}{\ensuremath{\abs{\GHZnodes{}}}}
\newcommand{\measuredsize}{\ensuremath{\abs{\measured{}}}}

\newcommand{\ie}{{i.e.}~} 
\newcommand{\eg}{{e.g.}~} 
\newcommand{\GHZtext}{\ensuremath{\mathrm{GHZ}}}

\begin{document}

\preprint{APS/123-QED}

\title{Extracting GHZ states from linear cluster states}

\author{J.\ de Jong$^{1}$, F.\ Hahn$^{1, 2}$, N.\ Tcholtchev$^{3}$, M.\ Hauswirth$^{1, 3}$, and A.\ Pappa$^{1, 3}$}
\affiliation{$^1$ Electrical Engineering and Computer Science Department, Technische Universit{\"a}t Berlin, 10587 Berlin, Germany}
\affiliation{$^2$ Dahlem Center for Complex Quantum Systems, Freie Universität Berlin, 14195 Berlin, Germany}
\affiliation{$^3$  Fraunhofer Institute for Open Communication Systems - FOKUS, 10589 Berlin, Germany}

\date{\today}

\begin{abstract}
\noindent 
Quantum information processing architectures typically only allow for nearest-neighbour entanglement creation. In many cases, this prevents the direct generation of \GHZtext{} states, which are commonly used for many communication and computation tasks. 
Here, we show how to obtain \GHZtext{} states between nodes in a network that are connected in a straight line, naturally allowing them to initially share linear cluster states.
We prove a strict upper bound of $\lfloor (n+3)/2 \rfloor$ on the size of the set of nodes sharing a \GHZtext{} state that can be obtained from a linear cluster state of $n$ qubits, using local Clifford unitaries, local Pauli measurements, and classical communication.
Furthermore, we completely characterize all selections of nodes below this threshold that can share a \GHZtext{} state obtained within this setting. 
Finally, we demonstrate these transformations on the \texttt{IBMQ Montreal} quantum device for linear cluster states of up to $n=19$ qubits.
\end{abstract}

\maketitle

\section{Introduction}\label{sec:introduction}
Recent years have seen exciting developments in quantum computation and communication, both in theory and experiment. Building upon the year-long research on bipartite settings, focus has now also turned towards multipartite settings, where multiple vertices in a network share quantum resources between them. While the correlations of Greenberger-Horne-Zeilinger (\GHZtext{}) states~\cite{GreenbergerBeyond1989} have naturally been the first to explore, other types of graph states~\cite{heinEntanglementGraphStates2006} have also been extensively examined
\cite{schlingemannErrorSyndromeCalculation2004,heinMultipartyEntanglementGraph2004,bouchetGraphicPresentationsIsotropic1988,hahnLimitationsNearestneighborQuantum2022}. The possible transformations between quantum states is a topic that is heavily studied \cite{briegelPersistentEntanglementArrays2001,vandennestGraphicalDescriptionAction2004,vandennestLocalUnitaryLocal2005,dahlberg2018transforming,dahlberg2020transform}, and while several hardness results have emerged \cite{dahlberg2020counting,Dahlberg2020transforminggraph}, a lot of practical questions remain unanswered. 

But why should we be interested in transforming one quantum state to another in the first place? One reason can be that it is not always possible to create the exact state that is necessary to perform a specific task, and we need to `retrieve' it  from some other state that is more practical to build.
For example, while the underlying network architecture might allow for nearest-neighbour interactions, it might not allow for the direct distribution of large \GHZtext{} states between distant parties. 
Here we show that there is an indirect remedy for this deficiency using suitable transformations of the distributed quantum states.
We focus on the transformation of \textit{linear cluster states}~\cite{briegelPersistentEntanglementArrays2001}, that arise naturally in linear networks. In particular we investigate the transformation to \GHZtext{} states, which are widely used in many quantum communication tasks including \textit{anonymous transmission}~\cite{christandlQuantumAnonymousTransmissions2005}, \textit{secret sharing}~\cite{HilleryQuantumSecretSharing1999,GHZ_Secret_sharing} and \textit{(anonymous) conference key agreement}~\cite{murtaQuantumConferenceKey2020,hahnAnonymousQuantumConference2020,grasselli2022secure}.

Such transformations require the removal of some of the qubits from the state by measuring them, such that only a selected subset of the qubits of the resource linear cluster state can in the end belong to the target \GHZtext{} state. We refer to these transformations as \GHZtext{} \textit{extractions}.
A previous study \cite{hahnQuantumNetworkRouting2019} showed how to extract three- and four-partite \GHZtext{} states from linear cluster states. Moreover, other works~\cite{frantzeskakis2022extracting,mannalathMultipartyEntanglementRouting2022} study a specific selection of the qubits of an odd-partite resource state.
Here, we conclude this study by providing a complete characterisation of which \GHZtext{} extractions are possible and which are not. Very importantly, we provide a tight upper bound to the size of the largest \GHZtext{} state that can be extracted, equal to $\lfloor (n+3)/2 \rfloor$; interestingly this is slightly higher than the bound of $n/2$ conjectured in Ref.~\cite{briegelPersistentEntanglementArrays2001} and than the sizes of the states extracted in the aforementioned studies~\cite{frantzeskakis2022extracting,mannalathMultipartyEntanglementRouting2022}.
In addition to our theoretical analysis, we perform demonstrations of implementations of the \GHZtext{} extractions from linear cluster states with  
$n \in \{5, 7, \dots, 19\}$
qubits on the \texttt{IBMQ Montreal} device.

Our manuscript is organized as follows: The notation, technical terminology and main definitions are introduced in Section~\ref{sec:notationterminology}. Section~\ref{sec:results} contains the main theoretical results. 
In Section~\ref{sec:demonstrations}, the demonstrations are introduced, discussed, and their results presented. 
Finally, Section~\ref{sec:discussion} discusses the obtained results and the opportunities for future research. 
The technical details are diverted to topical appendices: Appendix~\ref{app:proof_theorem} contains the proof of an important lemma stated in the theoretical section, Appendix~\ref{app:corrections} contains technical details regarding the post-processing steps during the extractions, and Appendix~\ref{app:fidelityestimate} contains technical details regarding the data analysis of the demonstrations section.

\begin{figure*}
    \centering
    \includegraphics[width=\linewidth, clip, trim=1cm 3cm 1cm 3.5cm]{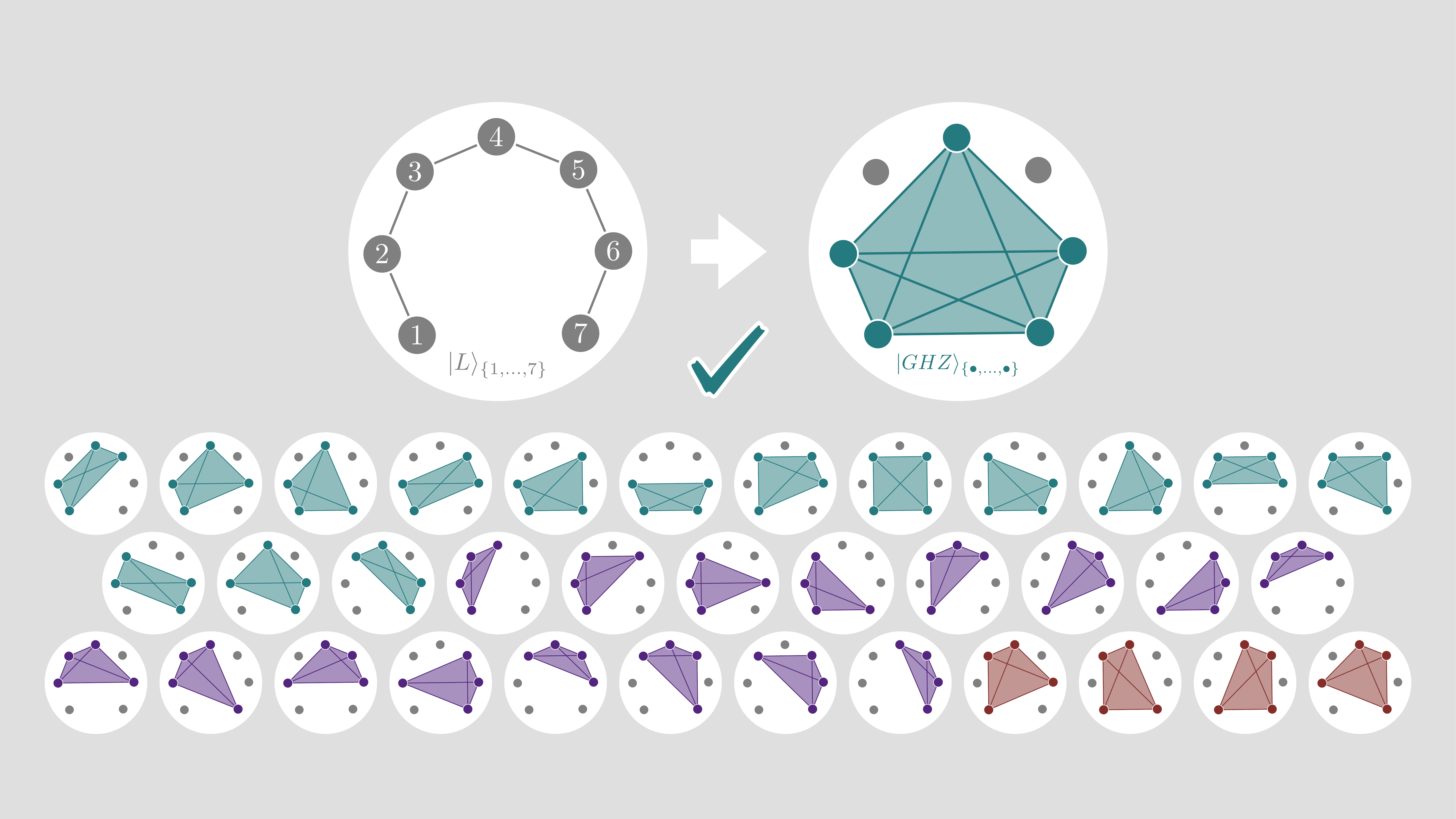}
    \caption{
    \textbf{Example of extracting \GHZtext{} states from a \textcolor{black}{linear cluster state with seven qubits}:}
    \textcolor{black}{The only $5$-partite \GHZtext{} state that can be extracted from this resource is on the qubits corresponding to $1,2,4,6,7$ and is highlighted in green}. 
    For $4$-partite \GHZtext{} states, we also \textcolor{black}{highlight all $15$ possible extraction patterns in green}, while \textcolor{black}{the patterns in brown are impossible due to Lemma~\ref{lem:2islands}} and \textcolor{black}{the patterns shown in violet are impossible due to both Corollary~\ref{cor:3island} and Lemma~\ref{lem:2islands}}.
    Note that due to Theorem~\ref{thm:bound} it is impossible to extract \GHZtext{} states with six or more qubits from this resource -- it is however trivially possible to extract all combinations of three-partite \GHZtext{} states.
    }
    \label{fig:network_overview}
\end{figure*}

\section{Notation and terminology}\label{sec:notationterminology}
In this work, two quantum graph states 
play a central role: 
we define \textit{linear cluster states} \Lin{} and \textit{GHZ states} \GHZ{} as
\begin{equation}
    \begin{split}
        \Lin{\{1, \dots, n\}} &\defeq 
        \frac{1}{2^{n}}\bigotimes_{i = 1}^{n}\left( \ket{0} + \ket{1}\sigma_{z}^{i+1} \right) \\
        \GHZ{\{1, \dots, m\}} &\defeq \frac{1}{\sqrt{2}} 
        \left(
        \bigotimes_{i = 1}^{m}\ket{0} 
        + 
        \bigotimes_{i = 1}^{m}\ket{1}\right)
    \end{split}
\end{equation}
and $\Lin{V}$, \GHZ{V} as the states corresponding to the vertex set $V$. When context permits, with \eg \Lin{n} we denote the linear cluster state of size $n$.

Our resource state is the $n$-partite linear cluster state $\Lin{V_L}$. 
As a graph state it corresponds to a line graph on the vertices $\network{} \defeq \{1,2,\ldots, n\}$.
Here, each vertex $\node{i}$ corresponds to the $i$-th qubit of $\Lin{V_L}$ and the edges of the graph correspond to nearest-neighbour entangling controlled phase gates. 
This structure allows us to use the terms \textit{left} and \textit{right} neighbours of \node{i} to indicate any vertices \node{h}, \node{j} with $h<i$, $i<j$, respectively; e.g.~the \textit{direct} left and right neighbours of \node{i} are $i\pm 1$.

Let $\GHZnodes{} \subset \network{}$ be a set of vertices for which we can extract a \GHZtext{} state from the linear cluster resource state.
Performing Pauli measurements on the qubits corresponding to $\measured \defeq \network \setminus \GHZnodes{}$, we obtain a post-measurement state which is local-Clifford equivalent to the \GHZ{\GHZnodes{}} state. 
By performing local operations based on the measurement outcomes, the state can then be locally transformed into this \GHZtext{} state.

This construction allows for \GHZnodes{} to inherit the neighbour structure from the linear network $\network{}$: For a vertex $j \in \GHZnodes{}$, we use $j_{-}$ and $j_{+}$ to indicate the left and right neighbour of $j$ in \GHZnodes{}, respectively.
We refer to the smallest and largest element of \GHZnodes{} as the \textit{boundaries} of the \GHZtext{} state. We finally define any selection of consecutive vertices $\node{i},\node{i+1}, \dots, \node{i+k} \in \GHZnodes{}$ as a \textit{$k$-island}.

\section{Main results}\label{sec:results}
We now examine what are the different types of \GHZtext{} states one can obtain from a given linear cluster state. We first provide an upper bound for the size \GHZsize{} of the extracted \GHZtext{} state, and we then show how to saturate it. In order to achieve this, we use Lemma~\ref{lem:2islands}, which provides an impossibility result for $2$-islands (the proof can be found in App.~\ref{app:proof_theorem}).

\begin{lemma}\label{lem:2islands}
    No $2$-island can have both a left and a right neighbour in \GHZnodes{}. If two vertices $\node{i}, \node{i+1}$ are in $\GHZnodes{}$, then there is either no vertex to the left of \node{i} or no vertex to the right of \node{i+1}.
\end{lemma}

Lemma~\ref{lem:2islands} implies that all vertices $\node{i}$ in the target \GHZtext~state must be `isolated' in the linear cluster state; $\node{i-1}$ and $\node{i+1}$ cannot be in $\GHZnodes{}$ (with the exception of the boundaries). A corollary for $3$-islands follows directly:
\begin{corollary}\label{cor:3island}
    If \GHZnodes{} contains a $3$-island, then $\GHZsize{}=3$.
\end{corollary}

\begin{proof}
     Let $\node{i}, \node{i+1}, \node{i+2}$ be a $3$-island in $\GHZnodes{}$ and assume that $|\GHZnodes{}|\geq 4$, \ie that we have $h<i$ or $j>i+2$ in \GHZnodes{}.
     This implies that either $\node{i},\node{i+1}$ form a $2$-island with both left-neighbour \node{h} and right-neighbour \node{i+2} or $\node{i+1},\node{i+2}$ form a $2$-island with both left-neighbour \node{i} and right-neighbour \node{j}.
    Both are in direct contradiction to Lemma~\ref{lem:2islands}.
\end{proof}
By the same argument, $k$-islands with $k \geq 4$ are impossible. Ultimately, such $k$-islands would contain $3$-islands in contradiction to Corollary~\ref{cor:3island}. Figure~\ref{fig:network_overview} illustrates examples.

This allows us to calculate the upper bound to $\GHZsize{}$.

\begin{theorem}\label{thm:bound}
The size of a \GHZtext{} state 
extractable
from an $n$-partite linear cluster state via local Clifford operations, local Pauli measurements, and local unitary corrections, is upper-bounded as $\GHZsize \leq \left\lfloor \frac{n+3}{2} \right\rfloor$.
\end{theorem}

\begin{proof}
    As there are at most two $2$-islands, for every other $\node{i}$ in $\GHZnodes{}$ both $\node{i \pm 1}$ were measured. 
    Thus, to maximize $\GHZsize$, we may have $\node{1}, \node{2}, \node{n-1},\node{n}$ in $\GHZnodes{}$, and $\measured{}$ containing every other vertex in between: For $n$ odd, $\measured{} = \{\node{3},\node{5},\dots,\node{n-2}\}$; for $n$ even $\measured{} = \{\node{3},\node{5},\dots,\node{n-5},\node{n-3},\node{n-2}\}$ \footnote{
        In the case of $n$ being even, there is more than one such pattern. While we have chosen here to measure the two consecutive vertices, \node{n-3} and \node{n-2}, other possibilities would have been to measure consecutive vertices further to the left and measure only the even vertices to the right. Another option would have been to measure not two consecutive vertices, but a vertex of one of the $2$ islands, \ie either \node{1},\node{2},\node{n-1} or \node{n}. It is important to note that all resulting sets $\measured{}$ have the same size. 
    }.
    In the even case, $\node{n-2}$ must be measured due to Corollary~\ref{cor:3island}.
    In both cases $\GHZsize{}=n-|\measured{}|$ is upper bounded by $\left\lfloor \frac{n+3}{2} \right\rfloor$.
\end{proof}

For example, the largest \GHZtext{} state that can be extracted from the $7$-qubit linear cluster state shown in Figure \ref{fig:network_overview} is the \GHZ{5} state where $\GHZnodes{} = \{1,2,4,6,7\}$ and $\measured{} = \{3,5\}$. 
Figure \ref{fig:network_overview} further shows all possible and impossible selections of \GHZnodes{} to extract the \GHZ{4} state.

We now show that there is a set of measurements that saturates the bound of Theorem~\ref{thm:bound} by explicitly giving such a measurement pattern.
For $n\leq 5$ this pattern was shown in Ref.~\cite{hahnQuantumNetworkRouting2019}.
For the general case, let us consider a case distinction with respect to the parity of $n$:

First, for odd $n$ we can choose $\measured=\{2i+1\}_{i = 1}^{\frac{n-3}{2}}$ and every corresponding qubit to be measured in the $\sigma_{x}$-basis; we refer to this measurement pattern as the \textit{maximal} pattern. Below, we show that this pattern actually gets the desired \GHZtext{} state.

The linear cluster state is a stabilizer state, \ie it is an element of the shared $+1$ eigenspace of the operators $\{l_{i} = \sigma_{z}^{i-1}\sigma_{x}^{i}\sigma_{z}^{i+1}\}_{i \in \network}$, where $\sigma_{z}^{0}$ and $\sigma_{z}^{n+1}$ are set equal to the identity. This set of operators forms the set of canonical \textit{generators} of an Abelian subgroup of the $n$-qubit Pauli group known as the \textit{stabilizer} of the linear cluster states. For an overview of the stabilizer formalism and stabilizer measurements in particular see \cite{GottesmanThesis1997},
\cite{GottesmanHeisenbergRepresentation1998}.

Consider the generator transformation 
\begin{equation}
    l_{2} \rightarrow l_{2}^{'} =  l_{2}l_{4} \dots l_{n-4}l_{n-2} \defeq \sigma_{z}^{1}\sigma_{z}^{n} \prod_{2i \in \network}\sigma_{x}^{2i},
\end{equation}
which ensures that $l^{'}_{2}$ and all odd-indexed generators commute with all measurement operators $\{\sigma_{x}^{j_{0}}\}_{j_{0}\in \measured{}}$. 
The post-measurement state is determined by replacing the other \measuredsize{} generators $\{l_{2i}\}_{i=2}^{\frac{n-1}{2}}$ with the measurement operators --together with a multiplicative phase depending on the respective measurement outcome. Then (after removing the support on the measured qubits and applying a Hadamard transformation to \node{1} and \node{n}) the post measurement state on \GHZnodes{} is characterized by the generators 
$\sigma_{x}^{\GHZnodes{}}$ and $\{m_{j_{0}}\sigma_{z}^{j_{0}} \sigma_{z}^{j_{0+}}\}_{j_{0} \in \GHZnodes{} \setminus \{n\}}$, where the $m_{j_{0}} = \pm 1$ are phases due to the measurement outcomes. These phases can be accounted for by applying $\sigma_{x}$-operations to a selection of the nodes, recovering the generators of the \GHZ{\GHZnodes{}} state.
The number of measurements implies $\GHZsize{} = n - \abs{\{3, 5,\dots,n-2\}} = \frac{n+3}{2}$ which saturates the bound for odd $n$.

Second, for even $n$ it suffices to observe that a $\sigma_{z}$-measurement on the qubit corresponding to \node{n} yields a linear cluster state \Lin{\{1,2,\dots, n-1\}} up to a randomized $\sigma_{z}^{n-1}$-correction depending on the measurement outcome. In analogy to the odd parity case we then obtain $\measured = \{3,5, \dots, n-3, n\}$ such that $\GHZsize{} = n - \measuredsize{} = \frac{n+2}{2} = \left\lfloor \frac{n+3}{2} \right\rfloor$ for even $n$.

Note that the even-case analysis above also applies for measuring an `internal' node in the $\sigma_{y}$-basis, rather than the first or last; this does introduce a Clifford rotation on the two neighbours of the node which needs to be accounted for~\cite{githublink}. The resulting state is then also LOCC equivalent to an ($n-1$)-partite linear cluster state on the remaining nodes, from which in turn a \GHZ{\frac{n+2}{2}} state can be extracted through the maximal pattern. This approach can be extended to more measurements, where additional `inside' nodes are measured in the $\sigma_{y}$-basis, and `outside' nodes are measured in the $\sigma_{z}$-basis. It is straightforward to see that any choice \GHZnodes{} allowed by Lemma~\ref{lem:2islands} can be seen as arising from such a setting. 

Finally, note that while Lemma~\ref{lem:2islands} does allow $2$-islands on the boundaries of the extracted \GHZtext{} states, they do not necessarily have to be contained in them. 
For example, \GHZ{\{1,3,5,7\}} can be extracted from \Lin{\{1,\ldots,7\}} as shown in Figure~\ref{fig:network_overview}. 
Rigorously stated, this pattern does not arise from one of the maximal patterns defined above, but can instead be considered as a maximal pattern \GHZ{\{0,1,3,5,7,8\}} extracted from \Lin{\{0,\ldots,8\}}.  
Here, the additional qubits corresponding to $0$ and $8$ are just ``virtual'' and not really there; they simply help visualize all possible patterns: 
\GHZ{\{1,3,5,7\}} can be extracted from the ``virtual'' state \GHZ{0,1,3,5,7,8} by measuring qubits $0$ and $8$ in the $\sigma_x$-basis.
The measurements on the other qubits are unaffected by this; the physical measurements of $2,4,6$ to obtain \GHZ{\{1,3,5,7\}} from \Lin{\{1,\ldots,7\}} are exactly the same as the ones that would be required to obtain \GHZ{\{0,1,3,5,7,8\}} from the ``virtual'' \Lin{\{0,\ldots,8\}}.
In this sense, all possible selections of \GHZnodes{} can be seen as subsets of the maximal measurement patterns defined above. 

These measurement patterns result in states LOCC equivalent to \GHZtext{} states; for explicit calculations of the necessary corrections to obtain the \GHZtext{} states themselves we refer to the supplementary material in \cite{githublink}.

\section{Demonstrations}\label{sec:demonstrations}
We used the \texttt{IBMQ Montreal} device to demonstrate our protocol for the maximal extraction of \GHZtext{} states from resource linear cluster states. For odd $n \in \{5, 7, \dots, 19\}$ we prepared the state
\begin{equation}\label{eq:rotatedlin}
    \ket{\psi}_{n} = \bigotimes_{i \text{ odd}}H^{i} \Lin{\{1,\ldots, n\}},
\end{equation}
\ie the linear cluster state with every odd qubit rotated to the $\sigma_{x}$-basis. 
We then extract \GHZtext{} states for $m \in \{\frac{n+3}{2}\}_{n} = \{4,5,\dots, 11\}$ using the maximal pattern described in the previous section.

\begin{figure}[ht]
	\includegraphics[width = 0.9\linewidth]{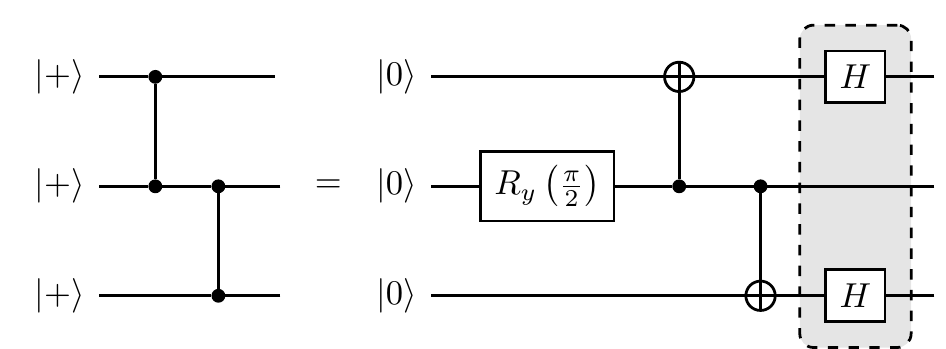}
	\caption{The circuit on the left and on the right are equal; the circuit on the left is the standard preparation of the \Lin{3} state. The generalization to \Lin{n} for higher (odd) $n$ follows naturally. The circuit on the right is the compiled version for the \texttt{IBMQ Montreal} chip. By not implementing the single-qubit gates in the grey box, the circuit depth is reduced by $25\%$. This results in a rotated linear cluster state $\ket{\psi}_{n}$ as defined in Eq.\eqref{eq:rotatedlin}. This state carries the same entanglement properties as the \Lin{n} state, and most notably can still be used to extract \GHZ{} states by adapting the measurement bases.
	}\label{fig:CZidentity}
\end{figure}

Implementing $\ket{\psi}_{n}$ instead of  \Lin{\{1,\ldots, n\}} allows us to reduce the circuit depth of the preparation circuit by one, when compiling for the gateset of the \texttt{IBMQ Montreal} device (Pauli-basis rotations, $CNOT$; see Figure~\ref{fig:CZidentity}).
When considering the \GHZtext{} extraction, this approach has further benefits; the necessary Hadamard transformations on the first and last qubit have, in essence, been applied `in advance', and the $\sigma_{x}$-measurements prescribed by the maximal pattern become $\sigma_{z}$-measurements, which are native to the device. The Pauli-based flips due to the measurement outcomes that are necessary to obtain the \GHZtext{} state can be performed in post-processing, as all the subsequent measurements on the \GHZtext{} state itself are in the Pauli basis.

To benchmark the results, we compute an estimate for the lower bound of the fidelity for both the linear cluster states and the \GHZtext{} states extracted from them. For the linear cluster states we use methods adapted from \cite{TiurevFidelitymeasurement2022} using insights originally presented in \cite{TothDetectingGenuine2005}; two measurement settings suffice to estimate the lower bound --one in which all qubits are measured in the $\sigma_{z}$-basis, and one in which all qubits are measured in the $\sigma_{x}$-basis. For the \GHZtext{} states we derive a similar technique. Again, two measurement settings suffice -- one where all the qubits of the \GHZtext{} state are measured in the $\sigma_{z}$-basis, and one where all the qubits of the \GHZtext{} are measured in the $\sigma_{x}$-basis. Both these measurement settings are performed in parallel to the $\sigma_{z}$-measurements of the qubits not included in the \GHZtext{} state that are required for the extraction. 
All measurements are repeated $32000$ times to calculate estimates for the expectation values. 

\begin{figure}[h!]
	\includegraphics[width = \linewidth]{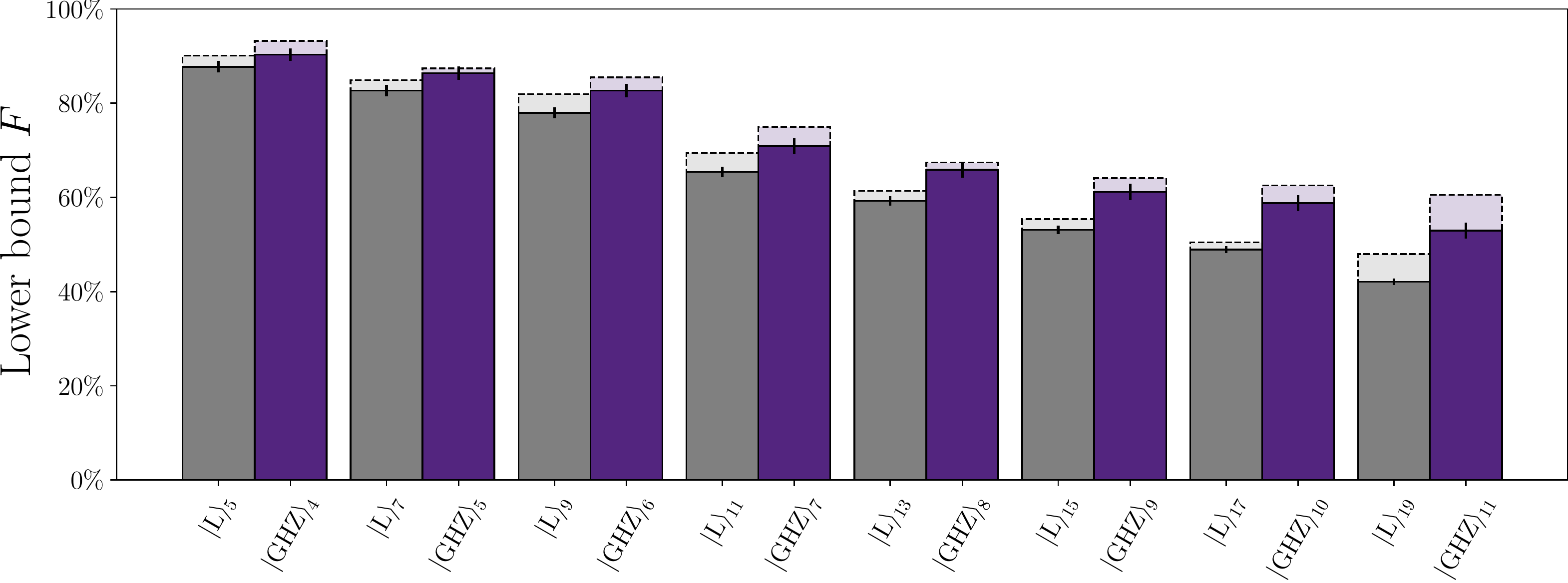}
	\caption{A lower bound of the fidelity of the rotated \textcolor{gray}{linear cluster} and \textcolor{violet}{\GHZtext{}} states prepared on the \texttt{IBMQ Montreal} device. 
  We prepared rotated (see Eq.~\ref{eq:rotatedlin}) linear cluster states $\ket{\psi}_{n}$   for $n \in \{5,7,\dots,19\}$ and extracted \GHZ{m} states for $m \in \{4,5,\dots, 11\}$
  using the maximal pattern introduced in Sec.~\ref{sec:results}. 
 For states with a higher number of qubits, the lower bound on the linear cluster state is increasingly worse due to a technical aspect of the estimation method (see App.~\ref{app:fidelityestimate} for details). The results are ordered such that each linear cluster state \Lin{n} is paired with the \GHZ{m} state extracted from it.
	}\label{fig:fidelitylowerbound}
\end{figure}

Figure~\ref{fig:fidelitylowerbound} shows the lower bounds on the fidelity for all linear cluster states that we generated with the \texttt{IBMQ Montreal} device --as well as for the \GHZtext{} states we extracted from them.
Note that our estimation method imposes a relative penalty for linear cluster state fidelity estimations compared to \GHZtext{} state fidelity estimations.
However, this does not mean that the fidelity of the \GHZtext{} states is truly higher than that of the linear cluster states from which they were extracted:
It simply means that we have used a method of bounding the fidelity from below, which works comparatively better for \GHZtext{} states than it does for linear cluster states. For the details of the estimation method we refer to App.~\ref{app:fidelityestimate}.


\section{Discussion}\label{sec:discussion} 
In this letter, we considered how to establish \GHZtext{} states between nodes that are share entanglement only with only a small number of nearest neighbours. In particular, given a linear cluster state shared between the nodes, we showed what are the possible \GHZtext{} states that can be obtained, the later being an indispensable resource in many quantum communication protocols including secret sharing \cite{GHZ_Secret_sharing}, electronic voting \cite{xue_simple_2017} and anonymous conference key agreement \cite{hahnAnonymousQuantumConference2020}.

Our results demonstrate that this process is possible but costly, since almost half of the linear cluster state qubits must be measured to obtain a \GHZtext{} state on the remaining qubits. Very importantly, we showed that there is in fact a tight upper bound to the size of the \GHZtext{} state we can obtain, higher than the one previously conjectured \cite{briegelPersistentEntanglementArrays2001}, thus solving a long-lasting open problem. We finally gave an exhaustive characterization of all possible \GHZtext{} states that can be extracted and provided a constructive method to obtain them, including the calculations for the necessary local rotations on the remaining qubits.

Our theoretical results are complemented by an implementation on IBM's superconducting quantum hardware, where this near-neighbour architecture is inherent. 
With fidelities of $80\%$ and higher for up to nine qubits, the results show that the generation of multi-partite entangled states is possible. 
It is also evident from the shown results that our method of extracting \GHZtext{} states does not compromise the fidelity of the target states compared to the resource states, since only local operations are required.  
Since the generation of linear cluster states can be, depending on the specific setting, experimentally more feasible, our approach shows a potentially more robust method of generating \GHZtext{} states.

Finally, extending our methods to other simple graph state resources does not seem trivial and requires further research. 
We note, however, that deriving an analogous characterisation for ring graph states, \ie graph states in which the leftmost and rightmost qubits of an otherwise linear cluster state are also connected, is straightforward using our methods.
In this case, only a single $2$-island is possible, so the upper bound for $\GHZsize{}$ becomes $\lfloor\frac{n+1}{2}\rfloor$. 

\section{Acknowledgements}\label{sec:acknowledgements}
A.P., J.d.J.~and F.H.~acknowledge support from the German Research Foundation (DFG, Emmy Noether Grant No. 418294583) and from the European Union via the Quantum Internet Alliance project. Upon completion of this work, we became aware of work similarly motivated~\cite{mannalathMultipartyEntanglementRouting2022,frantzeskakis2022extracting}.

\bibliography{Lin2GHZ_paper}{}

\providecommand{\noopsort}[1]{}\providecommand{\singleletter}[1]{#1}%
\begin{thebibliography}{29}%
\makeatletter
\providecommand \@ifxundefined [1]{%
 \@ifx{#1\undefined}
}%
\providecommand \@ifnum [1]{%
 \ifnum #1\expandafter \@firstoftwo
 \else \expandafter \@secondoftwo
 \fi
}%
\providecommand \@ifx [1]{%
 \ifx #1\expandafter \@firstoftwo
 \else \expandafter \@secondoftwo
 \fi
}%
\providecommand \natexlab [1]{#1}%
\providecommand \enquote  [1]{``#1''}%
\providecommand \bibnamefont  [1]{#1}%
\providecommand \bibfnamefont [1]{#1}%
\providecommand \citenamefont [1]{#1}%
\providecommand \href@noop [0]{\@secondoftwo}%
\providecommand \href [0]{\begingroup \@sanitize@url \@href}%
\providecommand \@href[1]{\@@startlink{#1}\@@href}%
\providecommand \@@href[1]{\endgroup#1\@@endlink}%
\providecommand \@sanitize@url [0]{\catcode `\\12\catcode `\$12\catcode
  `\&12\catcode `\#12\catcode `\^12\catcode `\_12\catcode `\%12\relax}%
\providecommand \@@startlink[1]{}%
\providecommand \@@endlink[0]{}%
\providecommand \url  [0]{\begingroup\@sanitize@url \@url }%
\providecommand \@url [1]{\endgroup\@href {#1}{\urlprefix }}%
\providecommand \urlprefix  [0]{URL }%
\providecommand \Eprint [0]{\href }%
\providecommand \doibase [0]{http://dx.doi.org/}%
\providecommand \selectlanguage [0]{\@gobble}%
\providecommand \bibinfo  [0]{\@secondoftwo}%
\providecommand \bibfield  [0]{\@secondoftwo}%
\providecommand \translation [1]{[#1]}%
\providecommand \BibitemOpen [0]{}%
\providecommand \bibitemStop [0]{}%
\providecommand \bibitemNoStop [0]{.\EOS\space}%
\providecommand \EOS [0]{\spacefactor3000\relax}%
\providecommand \BibitemShut  [1]{\csname bibitem#1\endcsname}%
\let\auto@bib@innerbib\@empty
\bibitem [{\citenamefont {Greenberger}\ \emph {et~al.}(1989)\citenamefont
  {Greenberger}, \citenamefont {Horne},\ and\ \citenamefont
  {Zeilinger}}]{GreenbergerBeyond1989}%
  \BibitemOpen
  \bibfield  {author} {\bibinfo {author} {\bibfnamefont {D.~M.}\ \bibnamefont
  {Greenberger}}, \bibinfo {author} {\bibfnamefont {M.~A.}\ \bibnamefont
  {Horne}}, \ and\ \bibinfo {author} {\bibfnamefont {A.}~\bibnamefont
  {Zeilinger}},\ }\enquote {\bibinfo {title} {Going beyond bell's theorem},}\
  in\ \href {\doibase 10.1007/978-94-017-0849-4_10} {\emph {\bibinfo
  {booktitle} {Bell's Theorem, Quantum Theory and Conceptions of the
  Universe}}},\ \bibinfo {editor} {edited by\ \bibinfo {editor} {\bibfnamefont
  {M.}~\bibnamefont {Kafatos}}}\ (\bibinfo  {publisher} {Springer
  Netherlands},\ \bibinfo {address} {Dordrecht},\ \bibinfo {year} {1989})\ pp.\
  \bibinfo {pages} {69--72}\BibitemShut {NoStop}%
\bibitem [{\citenamefont {Hein}\ \emph {et~al.}(2006)\citenamefont {Hein},
  \citenamefont {Dür}, \citenamefont {Eisert}, \citenamefont {Raussendorf},
  \citenamefont {Nest},\ and\ \citenamefont
  {Briegel}}]{heinEntanglementGraphStates2006}%
  \BibitemOpen
  \bibfield  {author} {\bibinfo {author} {\bibfnamefont {M.}~\bibnamefont
  {Hein}}, \bibinfo {author} {\bibfnamefont {W.}~\bibnamefont {Dür}}, \bibinfo
  {author} {\bibfnamefont {J.}~\bibnamefont {Eisert}}, \bibinfo {author}
  {\bibfnamefont {R.}~\bibnamefont {Raussendorf}}, \bibinfo {author}
  {\bibfnamefont {M.~V.~d.}\ \bibnamefont {Nest}}, \ and\ \bibinfo {author}
  {\bibfnamefont {H.-J.}\ \bibnamefont {Briegel}},\ }\href {\doibase
  10.48550/arXiv.quant-ph/0602096} {\emph {\bibinfo {title} {Entanglement in
  {Graph} {States} and its {Applications}}}},\ \bibinfo {type} {Tech. Rep.}\
  \bibinfo {number} {arXiv:quant-ph/0602096}\ (\bibinfo  {institution}
  {arXiv},\ \bibinfo {year} {2006})\ \bibinfo {note} {arXiv:quant-ph/0602096
  type: article}\BibitemShut {NoStop}%
\bibitem [{\citenamefont
  {Schlingemann}(2004)}]{schlingemannErrorSyndromeCalculation2004}%
  \BibitemOpen
  \bibfield  {author} {\bibinfo {author} {\bibfnamefont {D.-M.}\ \bibnamefont
  {Schlingemann}},\ }\href {\doibase 10.1063/1.1797533} {\bibfield  {journal}
  {\bibinfo  {journal} {Journal of Mathematical Physics}\ }\textbf {\bibinfo
  {volume} {45}},\ \bibinfo {pages} {4322} (\bibinfo {year}
  {2004})}\BibitemShut {NoStop}%
\bibitem [{\citenamefont {Hein}\ \emph {et~al.}(2004)\citenamefont {Hein},
  \citenamefont {Eisert},\ and\ \citenamefont
  {Briegel}}]{heinMultipartyEntanglementGraph2004}%
  \BibitemOpen
  \bibfield  {author} {\bibinfo {author} {\bibfnamefont {M.}~\bibnamefont
  {Hein}}, \bibinfo {author} {\bibfnamefont {J.}~\bibnamefont {Eisert}}, \ and\
  \bibinfo {author} {\bibfnamefont {H.~J.}\ \bibnamefont {Briegel}},\ }\href
  {\doibase 10.1103/PhysRevA.69.062311} {\bibfield  {journal} {\bibinfo
  {journal} {Physical Review A}\ }\textbf {\bibinfo {volume} {69}},\ \bibinfo
  {pages} {062311} (\bibinfo {year} {2004})}\BibitemShut {NoStop}%
\bibitem [{\citenamefont
  {Bouchet}(1988)}]{bouchetGraphicPresentationsIsotropic1988}%
  \BibitemOpen
  \bibfield  {author} {\bibinfo {author} {\bibfnamefont {A.}~\bibnamefont
  {Bouchet}},\ }\href {\doibase 10.1016/0095-8956(88)90055-X} {\bibfield
  {journal} {\bibinfo  {journal} {Journal of Combinatorial Theory, Series B}\
  }\textbf {\bibinfo {volume} {45}},\ \bibinfo {pages} {58} (\bibinfo {year}
  {1988})}\BibitemShut {NoStop}%
\bibitem [{\citenamefont {Hahn}\ \emph {et~al.}(2022)\citenamefont {Hahn},
  \citenamefont {Dahlberg}, \citenamefont {Eisert},\ and\ \citenamefont
  {Pappa}}]{hahnLimitationsNearestneighborQuantum2022}%
  \BibitemOpen
  \bibfield  {author} {\bibinfo {author} {\bibfnamefont {F.}~\bibnamefont
  {Hahn}}, \bibinfo {author} {\bibfnamefont {A.}~\bibnamefont {Dahlberg}},
  \bibinfo {author} {\bibfnamefont {J.}~\bibnamefont {Eisert}}, \ and\ \bibinfo
  {author} {\bibfnamefont {A.}~\bibnamefont {Pappa}},\ }\href {\doibase
  10.1103/PhysRevA.106.L010401} {\bibfield  {journal} {\bibinfo  {journal}
  {Phys. Rev. A}\ }\textbf {\bibinfo {volume} {106}},\ \bibinfo {pages}
  {L010401} (\bibinfo {year} {2022})}\BibitemShut {NoStop}%
\bibitem [{\citenamefont {Briegel}\ and\ \citenamefont
  {Raussendorf}(2001)}]{briegelPersistentEntanglementArrays2001}%
  \BibitemOpen
  \bibfield  {author} {\bibinfo {author} {\bibfnamefont {H.~J.}\ \bibnamefont
  {Briegel}}\ and\ \bibinfo {author} {\bibfnamefont {R.}~\bibnamefont
  {Raussendorf}},\ }\href {\doibase 10.1103/PhysRevLett.86.910} {\bibfield
  {journal} {\bibinfo  {journal} {Physical Review Letters}\ }\textbf {\bibinfo
  {volume} {86}},\ \bibinfo {pages} {910} (\bibinfo {year} {2001})}\BibitemShut
  {NoStop}%
\bibitem [{\citenamefont {Van~den Nest}\ \emph {et~al.}(2004)\citenamefont
  {Van~den Nest}, \citenamefont {Dehaene},\ and\ \citenamefont
  {De~Moor}}]{vandennestGraphicalDescriptionAction2004}%
  \BibitemOpen
  \bibfield  {author} {\bibinfo {author} {\bibfnamefont {M.}~\bibnamefont
  {Van~den Nest}}, \bibinfo {author} {\bibfnamefont {J.}~\bibnamefont
  {Dehaene}}, \ and\ \bibinfo {author} {\bibfnamefont {B.}~\bibnamefont
  {De~Moor}},\ }\href {\doibase 10.1103/PhysRevA.69.022316} {\bibfield
  {journal} {\bibinfo  {journal} {Physical Review A}\ }\textbf {\bibinfo
  {volume} {69}},\ \bibinfo {pages} {022316} (\bibinfo {year}
  {2004})}\BibitemShut {NoStop}%
\bibitem [{\citenamefont {Van~den Nest}\ \emph {et~al.}(2005)\citenamefont
  {Van~den Nest}, \citenamefont {Dehaene},\ and\ \citenamefont
  {De~Moor}}]{vandennestLocalUnitaryLocal2005}%
  \BibitemOpen
  \bibfield  {author} {\bibinfo {author} {\bibfnamefont {M.}~\bibnamefont
  {Van~den Nest}}, \bibinfo {author} {\bibfnamefont {J.}~\bibnamefont
  {Dehaene}}, \ and\ \bibinfo {author} {\bibfnamefont {B.}~\bibnamefont
  {De~Moor}},\ }\href {\doibase 10.1103/PhysRevA.71.062323} {\bibfield
  {journal} {\bibinfo  {journal} {Physical Review A}\ }\textbf {\bibinfo
  {volume} {71}},\ \bibinfo {pages} {062323} (\bibinfo {year}
  {2005})}\BibitemShut {NoStop}%
\bibitem [{\citenamefont {Dahlberg}\ and\ \citenamefont
  {Wehner}(2018)}]{dahlberg2018transforming}%
  \BibitemOpen
  \bibfield  {author} {\bibinfo {author} {\bibfnamefont {A.}~\bibnamefont
  {Dahlberg}}\ and\ \bibinfo {author} {\bibfnamefont {S.}~\bibnamefont
  {Wehner}},\ }\href@noop {} {\bibfield  {journal} {\bibinfo  {journal}
  {Philosophical Transactions of the Royal Society A: Mathematical, Physical
  and Engineering Sciences}\ }\textbf {\bibinfo {volume} {376}},\ \bibinfo
  {pages} {20170325} (\bibinfo {year} {2018})}\BibitemShut {NoStop}%
\bibitem [{\citenamefont {Dahlberg}\ \emph
  {et~al.}(2020{\natexlab{a}})\citenamefont {Dahlberg}, \citenamefont
  {Helsen},\ and\ \citenamefont {Wehner}}]{dahlberg2020transform}%
  \BibitemOpen
  \bibfield  {author} {\bibinfo {author} {\bibfnamefont {A.}~\bibnamefont
  {Dahlberg}}, \bibinfo {author} {\bibfnamefont {J.}~\bibnamefont {Helsen}}, \
  and\ \bibinfo {author} {\bibfnamefont {S.}~\bibnamefont {Wehner}},\
  }\href@noop {} {\bibfield  {journal} {\bibinfo  {journal} {Quantum Science
  and Technology}\ }\textbf {\bibinfo {volume} {5}},\ \bibinfo {pages} {045016}
  (\bibinfo {year} {2020}{\natexlab{a}})}\BibitemShut {NoStop}%
\bibitem [{\citenamefont {Dahlberg}\ \emph
  {et~al.}(2020{\natexlab{b}})\citenamefont {Dahlberg}, \citenamefont
  {Helsen},\ and\ \citenamefont {Wehner}}]{dahlberg2020counting}%
  \BibitemOpen
  \bibfield  {author} {\bibinfo {author} {\bibfnamefont {A.}~\bibnamefont
  {Dahlberg}}, \bibinfo {author} {\bibfnamefont {J.}~\bibnamefont {Helsen}}, \
  and\ \bibinfo {author} {\bibfnamefont {S.}~\bibnamefont {Wehner}},\
  }\href@noop {} {\bibfield  {journal} {\bibinfo  {journal} {Journal of
  Mathematical Physics}\ }\textbf {\bibinfo {volume} {61}},\ \bibinfo {pages}
  {022202} (\bibinfo {year} {2020}{\natexlab{b}})}\BibitemShut {NoStop}%
\bibitem [{\citenamefont {Dahlberg}\ \emph
  {et~al.}(2020{\natexlab{c}})\citenamefont {Dahlberg}, \citenamefont
  {Helsen},\ and\ \citenamefont {Wehner}}]{Dahlberg2020transforminggraph}%
  \BibitemOpen
  \bibfield  {author} {\bibinfo {author} {\bibfnamefont {A.}~\bibnamefont
  {Dahlberg}}, \bibinfo {author} {\bibfnamefont {J.}~\bibnamefont {Helsen}}, \
  and\ \bibinfo {author} {\bibfnamefont {S.}~\bibnamefont {Wehner}},\ }\href
  {\doibase 10.22331/q-2020-10-22-348} {\bibfield  {journal} {\bibinfo
  {journal} {{Quantum}}\ }\textbf {\bibinfo {volume} {4}},\ \bibinfo {pages}
  {348} (\bibinfo {year} {2020}{\natexlab{c}})}\BibitemShut {NoStop}%
\bibitem [{\citenamefont {Christandl}\ and\ \citenamefont
  {Wehner}(2005)}]{christandlQuantumAnonymousTransmissions2005}%
  \BibitemOpen
  \bibfield  {author} {\bibinfo {author} {\bibfnamefont {M.}~\bibnamefont
  {Christandl}}\ and\ \bibinfo {author} {\bibfnamefont {S.}~\bibnamefont
  {Wehner}},\ }in\ \href {\doibase 10.1007/11593447_12} {{\selectlanguage
  {english}\emph {\bibinfo {booktitle} {Advances in {Cryptology} - {ASIACRYPT}
  2005}}}},\ Vol.\ \bibinfo {volume} {3788},\ \bibinfo {editor} {edited by\
  \bibinfo {editor} {\bibfnamefont {D.}~\bibnamefont {Hutchison}}, \bibinfo
  {editor} {\bibfnamefont {T.}~\bibnamefont {Kanade}}, \bibinfo {editor}
  {\bibfnamefont {J.}~\bibnamefont {Kittler}}, \bibinfo {editor} {\bibfnamefont
  {J.~M.}\ \bibnamefont {Kleinberg}}, \bibinfo {editor} {\bibfnamefont
  {F.}~\bibnamefont {Mattern}}, \bibinfo {editor} {\bibfnamefont {J.~C.}\
  \bibnamefont {Mitchell}}, \bibinfo {editor} {\bibfnamefont {M.}~\bibnamefont
  {Naor}}, \bibinfo {editor} {\bibfnamefont {O.}~\bibnamefont {Nierstrasz}},
  \bibinfo {editor} {\bibfnamefont {C.}~\bibnamefont {Pandu~Rangan}}, \bibinfo
  {editor} {\bibfnamefont {B.}~\bibnamefont {Steffen}}, \bibinfo {editor}
  {\bibfnamefont {M.}~\bibnamefont {Sudan}}, \bibinfo {editor} {\bibfnamefont
  {D.}~\bibnamefont {Terzopoulos}}, \bibinfo {editor} {\bibfnamefont
  {D.}~\bibnamefont {Tygar}}, \bibinfo {editor} {\bibfnamefont {M.~Y.}\
  \bibnamefont {Vardi}}, \bibinfo {editor} {\bibfnamefont {G.}~\bibnamefont
  {Weikum}}, \ and\ \bibinfo {editor} {\bibfnamefont {B.}~\bibnamefont {Roy}}}\
  (\bibinfo  {publisher} {Springer Berlin Heidelberg},\ \bibinfo {address}
  {Berlin, Heidelberg},\ \bibinfo {year} {2005})\ pp.\ \bibinfo {pages}
  {217--235},\ \bibinfo {note} {series Title: Lecture Notes in Computer
  Science}\BibitemShut {NoStop}%
\bibitem [{\citenamefont {Hillery}\ \emph {et~al.}(1999)\citenamefont
  {Hillery}, \citenamefont {Bu\ifmmode~\check{z}\else \v{z}\fi{}ek},\ and\
  \citenamefont {Berthiaume}}]{HilleryQuantumSecretSharing1999}%
  \BibitemOpen
  \bibfield  {author} {\bibinfo {author} {\bibfnamefont {M.}~\bibnamefont
  {Hillery}}, \bibinfo {author} {\bibfnamefont {V.}~\bibnamefont
  {Bu\ifmmode~\check{z}\else \v{z}\fi{}ek}}, \ and\ \bibinfo {author}
  {\bibfnamefont {A.}~\bibnamefont {Berthiaume}},\ }\href {\doibase
  10.1103/PhysRevA.59.1829} {\bibfield  {journal} {\bibinfo  {journal} {Phys.
  Rev. A}\ }\textbf {\bibinfo {volume} {59}},\ \bibinfo {pages} {1829}
  (\bibinfo {year} {1999})}\BibitemShut {NoStop}%
\bibitem [{\citenamefont {Broadbent}\ \emph {et~al.}(2009)\citenamefont
  {Broadbent}, \citenamefont {Chouha},\ and\ \citenamefont
  {Tapp}}]{GHZ_Secret_sharing}%
  \BibitemOpen
  \bibfield  {author} {\bibinfo {author} {\bibfnamefont {A.}~\bibnamefont
  {Broadbent}}, \bibinfo {author} {\bibfnamefont {P.-R.}\ \bibnamefont
  {Chouha}}, \ and\ \bibinfo {author} {\bibfnamefont {A.}~\bibnamefont
  {Tapp}},\ }in\ \href {\doibase 10.1109/ICQNM.2009.20} {\emph {\bibinfo
  {booktitle} {2009 Third International Conference on Quantum, Nano and Micro
  Technologies}}}\ (\bibinfo {year} {2009})\ pp.\ \bibinfo {pages}
  {59--62}\BibitemShut {NoStop}%
\bibitem [{\citenamefont {Murta}\ \emph {et~al.}(2020)\citenamefont {Murta},
  \citenamefont {Grasselli}, \citenamefont {Kampermann},\ and\ \citenamefont
  {Bruß}}]{murtaQuantumConferenceKey2020}%
  \BibitemOpen
  \bibfield  {author} {\bibinfo {author} {\bibfnamefont {G.}~\bibnamefont
  {Murta}}, \bibinfo {author} {\bibfnamefont {F.}~\bibnamefont {Grasselli}},
  \bibinfo {author} {\bibfnamefont {H.}~\bibnamefont {Kampermann}}, \ and\
  \bibinfo {author} {\bibfnamefont {D.}~\bibnamefont {Bruß}},\ }\href
  {\doibase 10.1002/qute.202000025} {\bibfield  {journal} {\bibinfo  {journal}
  {Advanced Quantum Technologies}\ }\textbf {\bibinfo {volume} {3}},\ \bibinfo
  {pages} {2000025} (\bibinfo {year} {2020})},\ \bibinfo {note} {\_eprint:
  https://onlinelibrary.wiley.com/doi/pdf/10.1002/qute.202000025}\BibitemShut
  {NoStop}%
\bibitem [{\citenamefont {Hahn}\ \emph {et~al.}(2020)\citenamefont {Hahn},
  \citenamefont {de~Jong},\ and\ \citenamefont
  {Pappa}}]{hahnAnonymousQuantumConference2020}%
  \BibitemOpen
  \bibfield  {author} {\bibinfo {author} {\bibfnamefont {F.}~\bibnamefont
  {Hahn}}, \bibinfo {author} {\bibfnamefont {J.}~\bibnamefont {de~Jong}}, \
  and\ \bibinfo {author} {\bibfnamefont {A.}~\bibnamefont {Pappa}},\ }\href
  {\doibase 10.1103/PRXQuantum.1.020325} {\bibfield  {journal} {\bibinfo
  {journal} {PRX Quantum}\ }\textbf {\bibinfo {volume} {1}},\ \bibinfo {pages}
  {020325} (\bibinfo {year} {2020})}\BibitemShut {NoStop}%
\bibitem [{\citenamefont {Grasselli}\ \emph {et~al.}(2022)\citenamefont
  {Grasselli}, \citenamefont {Murta}, \citenamefont {de~Jong}, \citenamefont
  {Hahn}, \citenamefont {Bru{\ss}}, \citenamefont {Kampermann},\ and\
  \citenamefont {Pappa}}]{grasselli2022secure}%
  \BibitemOpen
  \bibfield  {author} {\bibinfo {author} {\bibfnamefont {F.}~\bibnamefont
  {Grasselli}}, \bibinfo {author} {\bibfnamefont {G.}~\bibnamefont {Murta}},
  \bibinfo {author} {\bibfnamefont {J.}~\bibnamefont {de~Jong}}, \bibinfo
  {author} {\bibfnamefont {F.}~\bibnamefont {Hahn}}, \bibinfo {author}
  {\bibfnamefont {D.}~\bibnamefont {Bru{\ss}}}, \bibinfo {author}
  {\bibfnamefont {H.}~\bibnamefont {Kampermann}}, \ and\ \bibinfo {author}
  {\bibfnamefont {A.}~\bibnamefont {Pappa}},\ }\href@noop {} {\bibfield
  {journal} {\bibinfo  {journal} {PRX Quantum}\ }\textbf {\bibinfo {volume}
  {3}},\ \bibinfo {pages} {040306} (\bibinfo {year} {2022})}\BibitemShut
  {NoStop}%
\bibitem [{\citenamefont {Hahn}\ \emph {et~al.}(2019)\citenamefont {Hahn},
  \citenamefont {Pappa},\ and\ \citenamefont
  {Eisert}}]{hahnQuantumNetworkRouting2019}%
  \BibitemOpen
  \bibfield  {author} {\bibinfo {author} {\bibfnamefont {F.}~\bibnamefont
  {Hahn}}, \bibinfo {author} {\bibfnamefont {A.}~\bibnamefont {Pappa}}, \ and\
  \bibinfo {author} {\bibfnamefont {J.}~\bibnamefont {Eisert}},\ }\href
  {\doibase 10.1038/s41534-019-0191-6} {\bibfield  {journal} {\bibinfo
  {journal} {npj Quantum Information}\ }\textbf {\bibinfo {volume} {5}},\
  \bibinfo {pages} {76} (\bibinfo {year} {2019})}\BibitemShut {NoStop}%
\bibitem [{\citenamefont {Frantzeskakis}\ \emph {et~al.}(2022)\citenamefont
  {Frantzeskakis}, \citenamefont {Liu}, \citenamefont {Raissi}, \citenamefont
  {Barnes},\ and\ \citenamefont {Economou}}]{frantzeskakis2022extracting}%
  \BibitemOpen
  \bibfield  {author} {\bibinfo {author} {\bibfnamefont {R.}~\bibnamefont
  {Frantzeskakis}}, \bibinfo {author} {\bibfnamefont {C.}~\bibnamefont {Liu}},
  \bibinfo {author} {\bibfnamefont {Z.}~\bibnamefont {Raissi}}, \bibinfo
  {author} {\bibfnamefont {E.}~\bibnamefont {Barnes}}, \ and\ \bibinfo {author}
  {\bibfnamefont {S.~E.}\ \bibnamefont {Economou}},\ }\href@noop {} {\bibfield
  {journal} {\bibinfo  {journal} {arXiv preprint arXiv:2203.07210}\ } (\bibinfo
  {year} {2022})}\BibitemShut {NoStop}%
\bibitem [{\citenamefont {Mannalath}\ and\ \citenamefont
  {Pathak}(2022)}]{mannalathMultipartyEntanglementRouting2022}%
  \BibitemOpen
  \bibfield  {author} {\bibinfo {author} {\bibfnamefont {V.}~\bibnamefont
  {Mannalath}}\ and\ \bibinfo {author} {\bibfnamefont {A.}~\bibnamefont
  {Pathak}},\ }\href {http://arxiv.org/abs/2211.06690} {\enquote {\bibinfo
  {title} {Multiparty {Entanglement} {Routing} in {Quantum} {Networks}},}\ }
  (\bibinfo {year} {2022}),\ \bibinfo {note} {arXiv:2211.06690 [math-ph,
  physics:quant-ph]}\BibitemShut {NoStop}%
\bibitem [{\citenamefont {Gottesman}(1997)}]{GottesmanThesis1997}%
  \BibitemOpen
  \bibfield  {author} {\bibinfo {author} {\bibfnamefont {D.}~\bibnamefont
  {Gottesman}},\ }\href {\doibase 10.48550/ARXIV.QUANT-PH/9705052} {\enquote
  {\bibinfo {title} {Stabilizer codes and quantum error correction},}\ }
  (\bibinfo {year} {1997})\BibitemShut {NoStop}%
\bibitem [{\citenamefont
  {Gottesman}(1998)}]{GottesmanHeisenbergRepresentation1998}%
  \BibitemOpen
  \bibfield  {author} {\bibinfo {author} {\bibfnamefont {D.}~\bibnamefont
  {Gottesman}},\ }\href {https://www.osti.gov/biblio/319738} {\enquote
  {\bibinfo {title} {The {H}eisenberg representation of quantum computers},}\ }
  (\bibinfo {year} {1998})\BibitemShut {NoStop}%
\bibitem [{git(2022)}]{githublink}%
  \BibitemOpen
  \href@noop {} {\enquote {\bibinfo {title} {{Supplementary information to the
  paper ``Extracting maximal entanglement from linear cluster states''}},}\
  }\bibinfo {howpublished}
  {\url{https://github.com/hahnfrederik/Extracting-maximal-entanglement-from-linear-cluster-states}}
  (\bibinfo {year} {2022}),\ \bibinfo {note} {[Online; accessed
  25-November-2022]}\BibitemShut {NoStop}%
\bibitem [{\citenamefont {Tiurev}\ and\ \citenamefont
  {S\o{}rensen}(2022)}]{TiurevFidelitymeasurement2022}%
  \BibitemOpen
  \bibfield  {author} {\bibinfo {author} {\bibfnamefont {K.}~\bibnamefont
  {Tiurev}}\ and\ \bibinfo {author} {\bibfnamefont {A.~S.}\ \bibnamefont
  {S\o{}rensen}},\ }\href {\doibase 10.1103/PhysRevResearch.4.033162}
  {\bibfield  {journal} {\bibinfo  {journal} {Phys. Rev. Research}\ }\textbf
  {\bibinfo {volume} {4}},\ \bibinfo {pages} {033162} (\bibinfo {year}
  {2022})}\BibitemShut {NoStop}%
\bibitem [{\citenamefont {T\'oth}\ and\ \citenamefont
  {G\"uhne}(2005)}]{TothDetectingGenuine2005}%
  \BibitemOpen
  \bibfield  {author} {\bibinfo {author} {\bibfnamefont {G.}~\bibnamefont
  {T\'oth}}\ and\ \bibinfo {author} {\bibfnamefont {O.}~\bibnamefont
  {G\"uhne}},\ }\href {\doibase 10.1103/PhysRevLett.94.060501} {\bibfield
  {journal} {\bibinfo  {journal} {Phys. Rev. Lett.}\ }\textbf {\bibinfo
  {volume} {94}},\ \bibinfo {pages} {060501} (\bibinfo {year}
  {2005})}\BibitemShut {NoStop}%
\bibitem [{\citenamefont {Xue}\ and\ \citenamefont
  {Zhang}(2017)}]{xue_simple_2017}%
  \BibitemOpen
  \bibfield  {author} {\bibinfo {author} {\bibfnamefont {P.}~\bibnamefont
  {Xue}}\ and\ \bibinfo {author} {\bibfnamefont {X.}~\bibnamefont {Zhang}},\
  }\href@noop {} {\bibfield  {journal} {\bibinfo  {journal} {Scientific
  Reports}\ }\textbf {\bibinfo {volume} {7}},\ \bibinfo {pages} {7586}
  (\bibinfo {year} {2017})}\BibitemShut {NoStop}%
\bibitem [{\citenamefont {Dehaene}\ and\ \citenamefont
  {De~Moor}(2003)}]{DehaeneCliffordgroup2003}%
  \BibitemOpen
  \bibfield  {author} {\bibinfo {author} {\bibfnamefont {J.}~\bibnamefont
  {Dehaene}}\ and\ \bibinfo {author} {\bibfnamefont {B.}~\bibnamefont
  {De~Moor}},\ }\href {\doibase 10.1103/PhysRevA.68.042318} {\bibfield
  {journal} {\bibinfo  {journal} {Phys. Rev. A}\ }\textbf {\bibinfo {volume}
  {68}},\ \bibinfo {pages} {042318} (\bibinfo {year} {2003})}\BibitemShut
  {NoStop}%
\end{thebibliography}%

\appendix
\section{Proof of Lemma \ref{lem:2islands}}\label{app:proof_theorem}
In this section we proof Lemma~\ref{lem:2islands}. We prove the theorem by contradiction. 
Fix a set \GHZnodes{} such that $\{i, i+1\} \subset \GHZnodes{}$ and let the post-measurement state $\ket{\psi}_{\GHZnodes{}}$ be locally equivalent to $\GHZ{\GHZnodes{}}$.  Assume now that there are both $j < i$ and $k > i+1$ for which both $j, k \in \GHZnodes{}$. W.l.o.g. assume that $j$ and $k$ are the direct left- and right neighbour of $i$ and $i+1$, respectively.

Recall that a set of generators for the linear cluster state is $\{l_{i_{0}} = \sigma_{z}^{i_{-}}\sigma_{x}^{i_{0}}\sigma_{z}^{i_{+}}\}_{i_{0} \in \network}$. If the post-measurement state is locally equivalent to the \GHZtext{} state then there must exist a (reversible) generator transformation such that their support on \node{i} and \node{i+1} coincides exactly with (the generators of) the \GHZtext{} state - up to local Clifford rotations. We will now show that, from a reversible transformation of the $\{l_{i_{0}}\}$, it is impossible to obtain such a set of generators when $j, i, i+1, k \in \GHZnodes{}$. This directly implies that a measurement pattern such that the \GHZtext{} state can be obtained is not possible. 

(A set of) generators for the \GHZtext{} state are, $\{\sigma_{x}^{\GHZnodes{}}\} \cup \{\sigma_{z}^{i_{0}}\sigma_{z}^{i_{+}}\}_{i_{0}\in \GHZnodes{}}$, where it is implied that $\sigma_{z}^{i_{+}} = 1$ when $i_{+} \not \in \GHZnodes{}$. Focusing on \node{i} and \node{i+1}, the only generators with non-trivial support are $\{\alpha, \beta, \gamma, \delta\} = \{\sigma_{a_{i}}^{i},  \sigma_{a_{i}}^{i} \sigma_{a_{i+1}}^{i+1}, \sigma_{a_{i+1}}^{i+1}, \sigma_{b_{i}}^{i} \sigma_{b_{i+1}}^{i+1}\}$, where $a_{i}, a_{i+1}, b_{i}, b_{i+1} \in \{x,y,z\}$ reflect the fact that the state is \textit{locally equivalent} to the \GHZtext{} state. This implies that $a_{i} \not = b_{i}$ and $a_{i+1} \not = b_{i+1}$.

Similarly, only the generators $l_{i-1}, l_{i}, l_{i+1}$ and $l_{i+2}$ of the linear cluster state (\ie those with support on \node{i} or \node{i+1}) can have a non-trivial contribution to the generator transformation on the vertices in question. Therefore, w.l.o.g., we can focus on just these four generators and restrict our attention to vertices \node{i} and \node{i+1}. If we show that there is no reversible transformation of $\{l_{k}\}_{k=\{i-1, i, i+1, i+2\}}$ to obtain $\{\alpha, \beta, \gamma, \delta\}$ when only considering these nodes, the lemma follows. We show there is no such transformation by exhaustive contradiction. 

There are three different ways of creating generator $\alpha$: \textbf{i)} $\alpha \propto l_{i-1} = \sigma_{z}^{i}$, \textbf{ii)} $\alpha \propto l_{i} \circ l_{i+2} = \sigma_{x}^{i}$, \textbf{iii)} $\alpha \propto l_{i-1} \circ l_{i} \circ l_{i+2} = \sigma_{y}^{i}$, where `$\alpha \propto l_{i-1}$' should be read as `$l_{i-1}$ takes the role of $\alpha$', and $\circ$ denotes the (qubit-wise) product (e.g. $l_{i} \circ l_{i+1} = \sigma_{x}^{i}\sigma_{z}^{i+1} \circ \sigma_{z}^{i}\sigma_{x}^{i+1} \hat{=} \sigma_{y}^{i}\sigma_{y}^{i+1}$, where the last equality is up to an irrelevant global phase). Similarly, there are three different ways of creating generator $\gamma$: \textbf{j)} $\gamma \propto l_{i+2} = \sigma+{z}^{i+2}$, \textbf{jj)} $\gamma \propto l_{i-1} \circ l_{i+1} = \sigma_{x}^{i+2}$, \textbf{jjj} $\gamma \propto l_{i-1} \circ l_{i+1} \circ l_{i+2} = \sigma_{y}^{i+2}$. Picking \eg \textbf{i)} and \textbf{j)} one sees that $\beta$ is fixed at $\propto \sigma_{z}^{i}\sigma_{z}^{i+1}$. But this is $l_{i-1} \circ l_{i+2} \propto \alpha \circ \gamma$, which would not be a reversible transformation of the generators $l_{i-1}, l_{i}, l_{i+1}$ and $l_{i+2}$.
Any other pair from $\{\textbf{i)},\textbf{ii)},\textbf{iii)}\}$ and $\{\textbf{j)},\textbf{jj)},\textbf{jjj)}\}$ would also necessitate such a non-reversible transformation.

In essence, when viewing the generators as vectors over $\mathbf{F}^{2n}_{2}$ through the binary representation \cite{DehaeneCliffordgroup2003}, the argument follows from the observation that (the vector associated with) $\beta$ lies in the subspace spanned by (the vectors associated with) $\alpha$ and $\gamma$. As such there can never be a reversible (\ie basis-preserving) operation on (the vectors associated with) $l_{i-1}, l_{i}, l_{i+1}$ and $l_{i+2}$ that obtains $\alpha, \beta$ and $\gamma$.

\section{Local-Clifford corrections to obtain \GHZtext{} states.}\label{app:corrections}
We provide a jupyter notebook for determining the required correction operations under 
\cite{githublink}. 

\section{Estimation of lower bound for fidelity in the demonstrations.}\label{app:fidelityestimate}
We here provide details for the method of estimation of the lower bound of the fidelity of both the linear cluster state and \GHZtext{} state, that has been used in the demonstrational implementation. The method is presented in and adapted from \cite{TiurevFidelitymeasurement2022} using insights originally presented in \cite{TothDetectingGenuine2005}.
The state that is prepared is
\begin{equation}\label{eq:rotatedlinapp}
    \ket{\psi}_{n} = \bigotimes_{i \text{ odd}}H^{i} \Lin{n},
\end{equation}
which is a linear cluster state rotated by Clifford operations and thus a stabilizer state.
Note that the generators $G^{L}$ for the stabilizer of $\ket{\psi}_{n}$ can be grouped into `odd' generators $G^{L}_{o} = \{\sigma_{z}^{i-1}\sigma_{z}^{i} \sigma_{z}^{i+1}\}_{i \text{ odd}}$ and `even' generators $G^{L}_{e} = \{ \sigma_{x}^{i-1} \sigma_{x}^{i} \sigma_{x}^{i+1} \}_{i \text{ even}}$, where again $\sigma_{z}^{0} = \sigma_{z}^{n+1} = 1$.
The fidelity of the prepared state $\rho$ with the rotated linear cluster state is $F(\rho, \ket{\psi}_{n}) = \tr \left[\rho \ketbra{\psi}{\psi}_{n}\right]$. 
Writing $G_{o(e)} = \prod_{g \in G^{L}_{o(e)}} \frac{\mathbb{I} + g}{2}$, 
and using $\ketbra{\psi}_{n} = \prod_{g \in G}\frac{\mathbb{I} + g}{2} = G_{o}G_{e}$, 
we can write 
\begin{align}
    F(\rho, \ket{\psi}_{n}) &= \tr\left[G_{o}G_{e}\rho\right] \\
    &= \tr\left[G_{o}\rho\right] + \tr\left[G_{e}\rho\right] - \tr\left[\mathbb{I} \rho \right] + \tr\left[K\rho\right],
\end{align}
where $K = \left(\mathbb{I} - G_{o}\right)\left(\mathbb{I} - G_{e}\right)$.
$K$ is positive semidefinite and thus we can discard the last term to obtain a lower bound for the fidelity:
\begin{equation}
	F(\rho, \ket{\psi_{n}}) \geq \mathbb{E}\left[G_{o}\right] + \mathbb{E}\left[G_{e}\right] - 1,
\end{equation}
where $\mathbb{E}\left[G_{o (e)}\right] = \frac{1}{2^{\abs{\mathcal{S}_{o (e)}}}}\sum_{\sigma \in \mathcal{S}_{o(e)}} \tr \left[\rho \sigma\right]$ with $\mathcal{S}_{o(e)} = \langle G^{L}_{o(e)}\rangle \subset \mathcal{S}$ the subgroup generated by the `odd' (`even') generators of $\ket{\psi}_{n}$.
Notably, all terms $\tr \left[\rho \sigma\right]$ comprise of only $\sigma_{z}$-basis ($\sigma \in \mathcal{S}_{o}$) or $\sigma_{x}$-basis ($\sigma \in \mathcal{S}_{e}$) measurements. This means that just two measurement settings suffice to estimate the lower bound: measuring all vertices in the $\sigma_{z}$-basis, and measuring all vertices in the $\sigma_{x}$-basis. By repeating these measurements $32000$ times and obtaining the outcome statistics, we estimate all terms $\tr \left[\rho \sigma\right]$ by selecting the outcomes associated with the $+1$ and $-1$ eigenspaces of all different observables.

For the \GHZtext{} state we use a similar method, where we now group the generators $G^{G}$ of the \GHZtext{} state into $G_{o}^{G} = \{\sigma_{x}^{\GHZnodes{}}\}$ and $G_{e}^{G} = \{\sigma_{z}^{j}\sigma_{z}^{j_{+}}\}_{j \in \GHZnodes{}}$, which again allows for an estimate of the lower bound with just two measurement settings. A caveat is that now there is only one `odd' generator and thus $\mathbb{E}\left[G^{G}_{o}\right] = \frac{1}{2}\tr\left[\rho \mathbb{I}\right] + \frac{1}{2}\tr\left[\rho \sigma_{x}^{\GHZnodes{}}\right]$. 
By definition $\tr\left[\rho \mathbb{I}\right] = 1$ and therefore the expectation value is more skewed towards $1$ than for the linear cluster state estimation. In other words it gives a higher bound on the fidelity when compared to the linear cluster state, since $G^{L}_{o} = O(2^{n})$ and as such the identity does not have such a strong impact on the estimate, especially for larger linear cluster states. To give another comparison between the two states, Figure~\ref{fig:fidelitylowerboundnoI} contains the same results as Figure~\ref{fig:fidelitylowerbound} from the main text, but with the identity-term omitted. This gives a lower but more equal estimate for both classes of states.

\begin{figure}[h]
	\includegraphics[width = \linewidth]{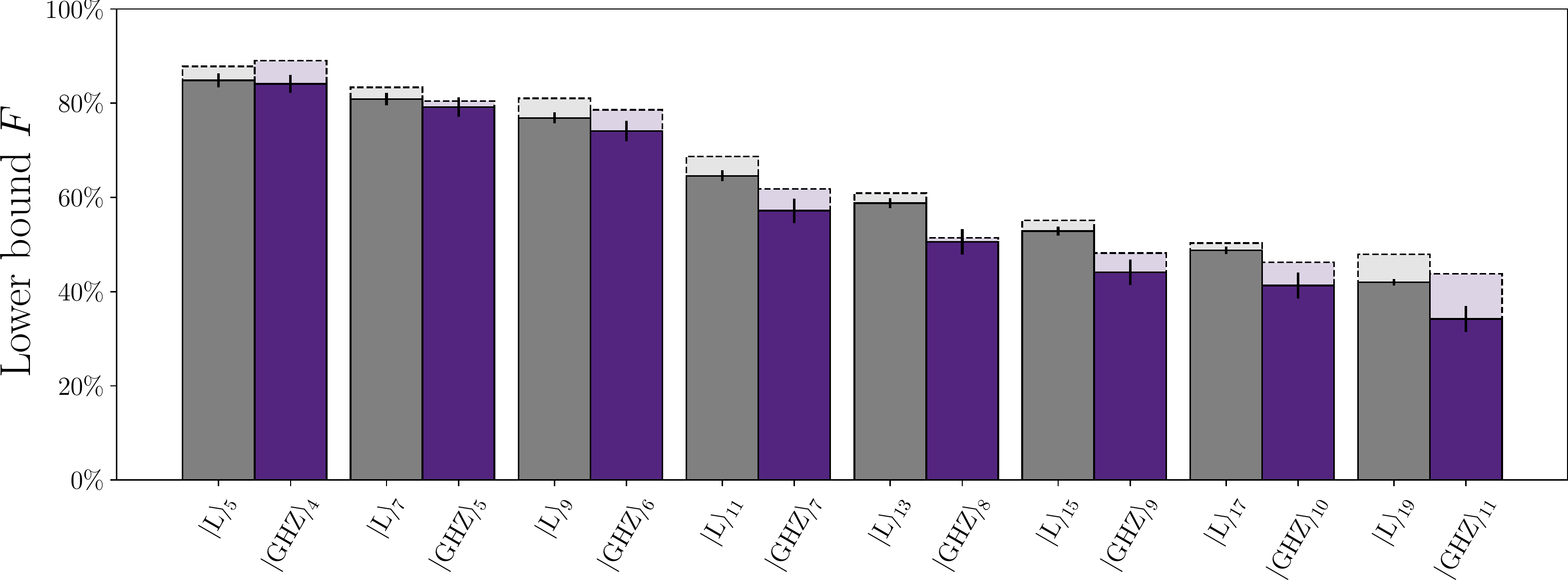}
	\caption{Lower bound on the fidelity using an adapted estimate method. In comparison with Figure~\ref{fig:fidelitylowerbound}, positive terms that favour the \GHZtext{} states are dropped, which renders a lower but more equal estimate on the fidelities for all states.
	}\label{fig:fidelitylowerboundnoI}
\end{figure}

\end{document}